\RequirePackage{amsmath}
\documentclass[runningheads,11pt]{llncs}
\usepackage{times}  
\usepackage{helvet}  
\usepackage{courier}  
\usepackage{url}  
\usepackage{graphicx}  
\usepackage{mathtools}
\usepackage[margin=1.3in]{geometry}

\usepackage{trunk/macros}

\begin{document}
%
\title{On the Necessary Memory to Compute the Plurality in Multi-Agent Systems}

%
%
\author{
Emanuele Natale\inst{1}
\and
Iliad Ramezani\inst{2}
}
\institute{%
Max Planck Institute for Informatics, Saarbr\"ucken, Germany
 \& Universit\'e C\^ote d'Azur, CNRS, I3S, Inria, France\\
\email{natale@i3s.unice.fr}
\and
Sharif University of Technology, Tehran, Iran\\
\email{iliramezani@ce.sharif.edu} 
}
\authorrunning{E. Natale et al.}
%

\maketitle
\begin{abstract}
    We consider the 
\emph{Relative-Majority Problem} (also known as \emph{Plurality}), 
in which,
given a multi-agent system where each agent is initially provided an input
value out of a set of $k$ possible ones, 
each agent is required to eventually compute the input value with
the highest frequency in the initial configuration.
We consider the problem in the general Population Protocols
model in which, given an underlying undirected
connected graph whose nodes represent the agents, edges are
selected by a \emph{globally fair} scheduler.  

The \emph{state complexity} that is required for solving the Plurality
Problem (i.e., the minimum number of memory states that each agent needs to
have in order to solve the problem), has been a long-standing open problem. 
The best protocol so far for the general multi-valued case requires
polynomial memory: Salehkaleybar et al. (2015) devised a
protocol that solves the problem by employing $\bigo(k 2^k)$ states per agent,
and they conjectured their upper bound to be optimal.  
On the other hand, under the strong assumption that agents initially agree
on a total ordering of the initial input values, G\k asieniec et al.
(2017),
provided
an elegant logarithmic-memory plurality protocol.

In this work, we refute Salehkaleybar et al.'s conjecture, by providing a
plurality protocol which employs $\bigo(k^{11}) $ states per agent.
Central to our result is an ordering protocol which allows to leverage on
the plurality protocol by G\k asieniec et al., of independent interest. 
We also provide a $\Omega(k^2)$-state lower bound on the necessary
memory to solve the problem, proving that the Plurality Problem cannot be
solved within the mere memory necessary to encode the output. 

\end{abstract}

\section{Introduction}

Consider a network of $n$ people, where each person supports one opinion from a
set of $k$ possible opinions. There is also a \emph{scheduler} who decides
in each round which pair of neighbors can interact. 
The goal is to \emph{eventually} reach an agreement on the opinion with the largest
number of supporters, i.e. the \emph{plurality} opinion (or \emph{majority}
when $k=2$).
Here, eventually means at an unspecified
moment in time, which the agents are not necessarily aware of (i.e. \emph{global termination} is not required \cite{santoro_design_2006}).

The main resource we are interested in minimizing is the 
\emph{state complexity} of each node: 

\emph{How many different states does each person need to
go through during such computation?}

This \emph{voting} task is known as the \emph{Plurality Problem} (or
as the \emph{Voting Problem}) in the asynchronous Population Protocols model
\cite{AADFP06,SSG15}.
For $k=2$, the problem is well understood: each person needs to maintain two bits in order for the people to elect the opinion of the majority~\cite{BTV09,MNRS17}, regardless of the network size $n$, and the problem cannot be solved with a single bit~\cite{MNRS17}. 

However, the state complexity of the problem for general $k$
has so far remained elusive: a clever protocol by Salehkaleybar et al.~\cite{SSG15},
called DMVR, shows how to solve the problem with $O(k2^k)$ states per person.
They conjectured the DMVR protocol to be optimal:

\smallskip
\begin{minipage}{0.95\linewidth}
\emph{``We conjecture that the DMVR protocol is an optimal
    solution for majority voting problem, i.e. at least $k 2^{k-1}$ states are required for any possible solution."}  
\vspace{5pt}
\end{minipage}

On the other hand, under the assumption that agents initially agree on a
total ordering of the initial input values, \cite{gasieniec_deterministic_2016}
provide an elegant plurality protocol which makes use of a polynomial number
of states only. 
It remained however rather unclear whether the above assumption can
be removed in order to achieve a polynomial number of states for the
general Plurality Problem as well. 

\subsection{Related Work}

Progress towards understanding the inherent computational complexity for a
multi-agent system to achieve certain tasks has been largely empirical in
nature. More recently, deeper insights have been offered by analytical studies
with respect to some coordination problems \cite{ranjbar__sahraei_theory_2014}. 
In this regard, understanding the amount of memory necessary for a multi-agent system in order
to solve a computational problem is a fundamental issue, as it constrains the
simplicity of the individual agents which make up the system
\cite{pitoni18}. 
Several research areas such as Chemical Reaction Networks \cite{Doty14} and Programmable
Matter \cite{gmyr_shape_2018} investigate the design of computing systems composed of elementary
units; in this regard, a high memory requirement for a computational problem
constitute a prohibitive barrier to its feasibility in such systems. 

The Plurality Problem (also known as Plurality \emph{Consensus} Problem in Distributed Computing),
is an extensively studied problem in many areas of distributed computing, such
as population protocols~\cite{AADFP06,BTV11,BTV09,MNRS17,SSG15}, fixed-volume
Chemical Reaction Networks~\cite{Doty14,TZB96}, asynchronous Gossip
protocols~\cite{BCNPS15,BCNPT16,BGPS06,GhaffariL17,GhaffariP16a}, Statistical
Physics \cite{LP08} and Mathematical Biology
\cite{BDDS14,Animals,MJTNS12,SUMPTER20081773}.

In the Population Protocols model, the memory is usually measured in terms of
the number of \emph{states} (state complexity) rather than the number of bits, 
following the convention for abstract automata \cite{holzer_descriptional_2011}. 
%
In the context of the Plurality Problem, for $k=2$, the protocols
of~\cite{BTV09,MNRS17} require $4$ states per node, and in~\cite{MNRS17}, they
showed that the problem cannot be solved with $3$ states. For general $k$, the
protocol of~\cite{SSG15} uses $O(2^{k-1}\cdot k)$ states per node, and the only
lower bound known has been so far the trivial $\Omega(k)$, as each node/agent
needs at least $k$ distinct states to specify its own opinion (which is from a
set of size $k$).
Under the crucial assumption that agents initially agree on a representation of
the input values as distinct integers, \cite{gasieniec_deterministic_2016}
provides an elegant solution to the Plurality Problem which employs $\bigo
(k^6)$ states only.

\subsection{Our Results}

In this work we refute the conjecture of~\cite{SSG15}, by devising a general
\emph{ordering} protocol which allows the agents to agree on a mapping of the
initial $k$ input values to the integers $\left\{ 0,\cdots,k-1 \right\}$,
thus satisfying the assumption of the protocol by
\cite{gasieniec_deterministic_2016}. 
We further show how to adapt the plurality protocol by
\cite{gasieniec_deterministic_2016} in a way that allows to couple its
execution in parallel with the ordering protocol such that, once the ordering
protocol has converged to the aforementioned mapping, the execution of the plurality protocol is also eventually
consistent with the provided ordering of colors. 
We emphasize that
agents are not required to detect when the protocol terminates; this is
indeed easily shown to be impossible under the general assumption of a fair
scheduler.
The resulting plurality protocol make use of $\bigo(k^{11})$ states
per agent. 
\begin{thm}
    \label{thm:generalupper}
    There is a population protocol $P_{general}$ which solves the Plurality
    Problem under a globally fair scheduler, by employing $\bigo(k^{11})$
    states per agent. 
\end{thm}
Furthermore, we prove that $k^2-k$ states per node are necessary (Theorem
\ref{thm:lower}).

\subsubsection{Insights on the Ordering Problem.}

The main idea for solving the ordering problem is to have some agents form a linked list, where each node is a single agent representing one of the initial colors. 
The fairness property of the scheduler allows for an \emph{adversarial} kind of asynchronicity in how agents' interactions take place.
Because of this distributed nature of the problem, (temporary) creation of multiple linked lists cannot be avoided. 
Thus, it is necessary to devise a way to eliminate multiple linked lists, whenever more than one of them are detected. 
We achieve this goal by having agents from one of the linked lists  \emph{leave} it; also, as soon as these leaving agents interact with their successor or predecessor in their former list, they force them to leave the list as well, thus propagating the removal process until the entire list gets destroyed. 

On the other hand, in order to form the linked list, the simple idea of having \emph{removed} agents appending themselves to an existing linked list does not work. One of the issues with this naive approach is that a free agent $u$ may interact with the last agent $v$ of a list $\ell$ which is in the process of being destroyed, but the removal process in $\ell$ may still not have reached $v$. 
Our approach to resolve this latter issue consists of, firstly, forcing the destruction process of a linked lists to start from the first agents of the lists, and secondly, forcing free agents to attach to an existing list by \emph{climbing it up} from its first agent and appending themselves to its end once they have traversed it all. This way, by the time that there is only one \emph{first agent} $r$ of a linked list (we call such agents \emph{root} agents), we can be sure that all the free agents must follow the linked list starting by the agent $r$, thus avoid extending incomplete linked lists.

\subsection{Model and Basic Definitions}

\subsubsection{Population Protocols} In this work, we consider the communication
model of Populations Protocols~\cite{AADFP06}: the multi-agent system is
represented by a connected graph
$G=(V,E)$ of $n$ nodes/agents,
where each node implements a finite state machine with state
space $\Sigma$. The communication in this model proceeds in discrete steps. 
We remark that, as for asynchronous continuous-time models with
    Poisson transition rates, they can always be mapped to a discrete-time
    model \cite{EFKMT16}.

At each time step, an (oriented) edge is chosen by a certain
\emph{scheduler}, and the two endpoint nodes interact. Furthermore, there is a transition function
$\Gamma:\Sigma\times\Sigma\rightarrow\Sigma\times\Sigma$ that, given an ordered
pair of states $\left(\sigma_{u},\sigma_{v}\right)\in\Sigma\times\Sigma$ for
two interacting nodes $u$ and $v$, returns their new states
$\Gamma\left(\left(\sigma_{u},\sigma_{v}\right)\right)
=\left(\sigma_{u}^{\prime},\sigma_{v}^{\prime}\right)$. We call
\emph{configuration}, and denote it by
$\conf^{(t)}$, the vector whose entry $u$ corresponds to agent $u$'s state
after $t$ time steps. We say
that a configuration $\conf'_1$ is \emph{reachable} from configuration
$\conf'_2$ if there exists a sequence of edges $\edgeseq$ such that if we start from
$\conf'_2$ and we let the nodes interact according to $\edgeseq$, the resulting
configuration is $\conf'_1$.

In recent works, the scheduler in this model is typically assumed to be probabilistic: the edge that is 
selected at each step is determined by a probability distribution on the edges. 
The most general studied scheduler is the \emph{fair scheduler} \cite{AAER07},
which guarantees the following \emph{global fairness property} \cite{ABBS17,BBCS15}.
\begin{defin}
    \label{def:fairness}
    A scheduler is said to be \emph{globally fair}, iff whenever a configuration $\conf$
    appears infinitely often in an infinite execution $\conf^{(1)},
    \conf^{(2)}, \cdots$, also any configuration $\conf'$ reachable from
    $\conf$ appears infinitely often. 
\end{defin} 
Some of our results hold for an even weaker\footnote{Formally, the globally fair scheduler is not a special case of the weak one since, if the activation of an edge does not lead to a different configuration, it can be ignored under a globally fair scheduler. However, if such \emph{useless} activations are ignored, it is easy to see that the globally fair scheduler is a special case of the weak one.} version of scheduler, which satisfies the
\emph{weak fairness property} \cite{BBCS15,gasieniec_deterministic_2016}. 
\begin{defin}
    \label{def:weak}
    A scheduler is said to be \emph{weakly fair}, 
    iff any edge $e \in E$ appears infinitely often in the activation series $e_1, e_2, ... $ .
\end{defin}
%
    Note that any probabilistic scheduler which selects any edge with
    a positive probability, is a globally fair scheduler, in the sense that the
    global fairness property holds with probability 1. 
%
Indeed, the fairness condition for a scheduler may be viewed as an attempt to
capture, in a general way, useful probability-1 properties in a
probability-free model \cite{AAER07}. This is crucially the case when
correctness is required to be deterministic (i.e. the probability of failure
should be 0) \cite{MNRS17,SSG15}. 

We emphasize that our theoretical results concern the \emph{existence} of certain times in the execution of the protocols for which some given properties hold, 
but no general time upper bound is provided, since a fair scheduler can typically delay some edge activation arbitrarily.

\subsubsection{$k$-Plurality Problem.}
Let $G=(V,E)$ be a network of $n$ agents, such that each agent
$v\in V$ initially supports a value in a set of possible values
$C$ of size $k$. We refer to the $k$ input values as \emph{colors}. 
For each
color $c\in C$, denote by $supp(c)$ the set of agents supporting color $c$. 
We further denote $\col v$ as the input color of $v \in V$.
We say that a population protocol solves the $k$-plurality problem if
it reaches any configuration $\conf^{(t)}$, such that for any $t'\geq t$ it
holds that the agents agree on the color with the greatest number of
supporters in the initial configuration $\conf^{(0)}$. 
More formally, there is an \emph{output function} $\Phi:\Sigma \rightarrow C$ such that for 
any $ t'\geq t$ and any agent $u$, $\Phi((\conf^{(t')})_u)$ equals the plurality color. 
If the relative majority is not unique, the agents should
reach agreement on any of the plurality colors.

In this work, we focus on solving the $k$-Plurality
Problem under a fair scheduler with the goal
	of optimizing the state complexity, 
which we denote by $\memory k$.

We emphasize that we do not assume any
non-trivial lower bounds on the support of the initial majority compared to other colors, 
nor that the agents know the size of the network $n$, or that they
know in advance the number of colors $k$. 
We do not make any assumption on the underlying graph other than
 connectedness. 
We remark that the analysis of our protocol $P_{general}$ in Theorem \ref{thm:generalupper} 
holds for strongly connected directed graphs; 
however, for the sake of simplicity, 
we restrict ourselves to the original setting by \cite{ranjbar__sahraei_theory_2014}. 

Crucially, motivated by real-world scenarios such as DNA computing and biological
protocols, we do not even assume that the nodes initially agree on a binary
representations of the colors: they are only able to recognize
whether two colors are equal and to memorize them. 
This latter assumption separates the polynomial state complexity 
of \cite{gasieniec_deterministic_2016} from the exponential state complexity of \cite{ranjbar__sahraei_theory_2014}.

\section{Lower Bound on $\memory k$}

Since the agents need at least to be able to distinguish their initial colors from each other, 
the trivial lower bound $\memory k \geq k$ follows. 
In this section, we show that $\memory k \in \Omega(k^2)$.
\begin{thm}
    \label{thm:lower}
    Any protocol for the $k$-Plurality Problem 
    requires at least $k^2-k$ memory states per agent.
\end{thm}
\begin{proof}
    The high level idea is to employ an indistinguishability argument. That is,
    we show that for any protocol with less than $k^2-k$ states, there must be
    two initial configurations, $\initconf_1$ and $\initconf_2$, with different plurality colors,
    such that a configuration is reachable from both $\initconf_1$
    and $\initconf_2$. Therefore, the protocol must fail in at least one of these two
    initial configurations.
    
    Let $P$ be a protocol that solves the plurality consensus problem with $k$
    initial colors, and let $\Phi : \Sigma \rightarrow C$ be the output function of $P$.
    Define $\finalstates c = \{ \sigma
    \in \Sigma \,|\, \Phi(\sigma) = c\}$. 
    We start by observing that there must be some color $c^* \in C$, such that $|\finalstates{c^*}| \leq |\Sigma| / k$.
    For any initial configuration
    $\initconf$ and color $c$, let $\Delta^{\initconf}_{c}$ be the number of agents in $\initconf$ with
    initial color $c$. 
    \begin{defin}
        For an odd integer $x > 0$, let $\initconfs{c^*}{x}$ be the set of all
        initial configurations $\initconf$, such that $|\initconf| = 2x-1$, $\Delta^{\initconf}_{c^*} = x$
        and for any color $c \neq c^*$, $\Delta^{\initconf}_{c}$ is an even number.
    \end{defin}
    Given that, for the sake of the lower bound,
    we can assume a complete topology, the number of configurations in $\initconfs{c^*}{x}$ is equal to the
    number of ways to put $(x-1)/2$ pair of balls into $k-1$ bins. 
    Therefore, we have
    $
        |\initconfs{c^*}{x}| \geq ( \nicefrac{\frac{x-1}{2}+k-2}{k-2})^{k-2} .
    $
    For each $\initconf \in \initconfs{c^*}{x}$, since the plurality color in $\initconf$ is
    $c^*$, $\initconf$ will reach a configuration that $\Phi$ maps all agents in the configuration to $c^*$. 
    The number of such possible configurations is at most the number of ways to put
    $2x-1$ balls into $|\finalstates{c^*}|$ bins. For a sufficiently large $x$,
    the number of such possible configurations is at most
    $
        (\nicefrac{(2x-1+\frac{|\Sigma|}{k}-1)e}{\frac{|\Sigma|}{k}-1})^{\frac{|\Sigma|}{k}-1}.
    $
    Observe that for $|\Sigma| < k^2-k$ and sufficiently large $x$, the upper
    bound on the number of possible final configurations is less than the lower
    bound on $|\initconfs{c^*}{x}|$. Therefore, there must be two distinct initial configurations 
    $\initconf_1, \initconf_2 \in \initconfs{c^*}{x}$ and a configuration $\conf$ in which all agents are
    mapped to $c^*$, such that $\conf$ is reachable from both $\initconf_1$ and $\initconf_2$, by some activation sequences $T_1$ and $T_2$ respectively.
    By definition of $\initconfs{c^*}{x}$, we have the following observation.
    \begin{observation}
        For each $\initconf, \initconf' \in \initconfs{c^*}{x}$ where $\initconf \neq \initconf' $, there exists a color
        $c$ such that $|\Delta^{\initconf}_{c} - \Delta^{\initconf'}_{c}| \geq 2$.
    \end{observation}
    Let $c$ be the color obtained from Observation 2 when applied to $\initconf_1$ and
    $\initconf_2$. Without loss of generality, assume that $\Delta^{\initconf_1}_{c} \geq
    \Delta^{\initconf_2}_{c} + 2$. Let $\initconf_3$ be an initial configuration with $x -
    \Delta^{\initconf_1}_{c} + 1$ agents, all having initial color $c$. Let us
    consider the two initial configurations $\initconf_4 = \initconf_1 \cup \initconf_3$ and
    $\initconf_5 = \initconf_2 \cup \initconf_3$. Observe that the plurality color in $\initconf_5$ is
    still $c^*$, while the plurality color in $\initconf_4$ is now $c$. Since $T_1$
    and $T_2$ are possible initial sequences of interactions in $\initconf_4$
    and $\initconf_5$ respectively, both $\initconf_4$ and $\initconf_5$ can reach the
    configuration $\conf \cup \initconf_3$. Therefore, a protocol $P$ using only $k^2-k-1$ states can fail
    to distinguish between initial configurations $\initconf_4$ and $\initconf_5$. Hence, $P$ fails to solve the problem on at
    least one initial configuration.
\end{proof}

\section{Upper bound on $\memory k$}

In the following, we present a protocol that solves the problem
with polynomial state complexity; we prove that $\memory k \in O(k^{11})$. 
The
protocol proposed by G\k asieniec et al. [2] solves the problem using a polynomial number of states, under the hypothesis that agents agree on a way to represent each color with a $m$-bit label. 

First, we  present a protocol that constructs such a shared labeling for the input colors (Theorem \ref{thm:ordering}). Then, we combine these two protocols to design a new protocol that solves the $k$-Plurality Problem (Theorem \ref{thm:generalupper}).

\subsection{Protocol for the Ordering Problem}

In the Ordering Problem, each agent $a \in V$ initially obtains its input color $c_a$, from a set of possible colors $C$ of size $k$. 
The goal of the agent is to eventually agree on a bijection between the set of the possible input colors of size $k$, and the integers $\{0,...,k-1\}$. In other words, each agent $a $ eventually gets a \emph{label} $d_a \in \{0, 1, ..., k-1\}$ , such that for any two agents $a$ and $b$, $d_a = d_b$ iff $c_a = c_b$. 
We want to solve the Ordering Problem by means of a protocol which uses as few states as possible.

A weakly fair scheduler activates pairs of agents to interact. 
%
%
We consider the underlying topology of possible interactions to be a complete directed graph. 
We show how to remove such assumption in General Graphs section.

In this section, we prove the following theorem. 
\begin{thm}
    \label{thm:ordering}
    There is a population protocol $P_o$ which solves the Ordering Problem under a weakly fair scheduler, by employing $\bigo(k^4)$
    states per agent. 
\end{thm}

We refer the reader to the section Insights on the Difficulty in the Introduction for an overview of the main ideas behind protocol $P_o$.  

\textbf{Memory Organization.}
The state of each agent $a$, encodes the following information:
\begin{enumerate}
\item $c_a$, the initial color, which never changes.
\item $d_a$, the desired value, stored in $\lceil{log_2k}\rceil$ bits.
\item $l_a$, a bit, indicating whether or not $a$ is a leader.
\item $r_a$, a bit, indicating whether or not $a$ is a root.
\item $pre_a$, a color from the set $C$. If $r_a = 0$ and $a$ is on a linked list, then $pre_a$ is the color of the agent preceding $a$ on the linked list. Otherwise $pre_a$ is set to be $c_a$.
\item $suc_a$, a color from the set $C$. If $a$ is on a linked list, $suc_a$ is the color of the agent succeeding $a$ on the linked list (or $c_a$ if $a$ is the last agent in the linked list). Otherwise, $suc_a$ is the color of the agent whom $a$ is following on a linked list, to reach the end of that linked list, or $c_a$ if $a$ is not following a linked list yet.
\end{enumerate}
Thus, the number of states used is at most $8k^4$.

\textbf{Definitions.}
An agent $a$ is called a \emph{leader}, iff $l_a$ is set.
A leader $a$ is called a \emph{root}, iff $r_a$ is set.
A leader $a$ is called \emph{isolated}, iff $a$ is not a root and $pre_a = c_a$.

A linked list of $n$ links, is a sequence of leaders $a_0, a_1, .., a_{n}$, such that only $a_0$ is a root, and $\forall i, 0 < i \leq n: suc_{a_{i-1}} = c_{a_i} \land pre_{a_i} = c_{a_{i-1}}$. 
A linked list is said to be \emph{consistent}, iff none of its agents' information change by any sequence of further activations, except possibly $suc_a$ where $a$ is the last agent on the linked list.

An isolated agent $a$ is a \emph{good} agent, iff $suc_a$ is either $c_a$ or the color of one of the agents of a consistent linked list.

\textbf{Initialization.}
Before the execution of the protocol, each agent sets $d = 0$, $l = 1$, $r = 1$, $pre = c$ and $suc = c$.

\textbf{Transition Function.}
Let us suppose two agents $a, b \in A$ interact, $a \neq b$. The transition function $\Gamma_o$ that updates their states is given by the following Python code, where clear function is for \emph{isolating} an agent. 

\begin{lstlisting}[
    label={alg:ordering},
    caption={Protocol $P_o$ for the Ordering  Problem.},
    captionpos=b,
    language=Python,
    escapechar=@, 
    %showstringspaces=false,
    %tabsize=1,
    %numbers=left,
    breaklines=true,
    basicstyle=\small,
    ]
def clear(u):
   r_u = d_u = 0
   pre_u = suc_u = c_u

if c_a == c_b:
   if l_a and l_b and (not r_a or r_b):
      l_a = r_a = 0   #1
   elif l_a and not l_b:
      d_b = d_a   #2
elif c_a != c_b and l_a and l_b:
   if r_a and r_b:
      clear(a)   #3
   elif not r_a and pre_a == c_a and suc_a == c_b and pre_b == c_b:
      clear(a)   #4
   elif r_a and not r_b:
      if suc_a == c_a or (suc_a == c_b and (pre_b != c_a or d_b != 1)):
         d_b = 1   #5 @\label{rule:5}@
         suc_a = suc_b = c_b
         pre_b = c_a
      elif pre_a == c_a:
         suc_b = suc_a   #6
   elif not r_a and not r_b:
      if pre_a != c_a and suc_a == c_b and pre_b != c_a:
         suc_a = c_a   #7
      elif pre_b == c_a and (pre_a == c_a or suc_a != c_b):
         clear(b)   #8
      elif pre_a != c_a and suc_a == c_b and pre_b == c_a and d_a + 1 != d_b:
         suc_a = c_a   #9
         clear(b)
      elif pre_a == c_a and suc_a == c_b:
         if suc_b != c_b:
            suc_a = suc_b   #10
         else:
            d_a = d_b + 1   #11
            suc_b = suc_a = c_a
            pre_a = c_b
\end{lstlisting}

As seen above, there are 11 rules. The rules are defined for directed pair interactions, but can easily be modified to handle the undirected-interaction case.

\subsubsection{Proof of Theorem \ref{thm:ordering}.} 

We now prove the correctness of Protocol $P_o$ (Algorithm \ref{alg:ordering}). 
We have the following. 

\begin{lemma}
    \label{lemma:orderinglead}
    After some number of activations $T$, in each nonempty set $supp(c)$ of agents, only one is a leader, and among all leaders only one is a root. After such configuration is reached, the leader and root bits of all agents will never change.
\end{lemma}
\begin{proof}
    The protocol never changes a leader or root bit from False to True.
    When two leaders with the same color interact, one of them clears its leader bit, due to Rule 1 (notice that the direction of interaction is relevant here).
    Therefore, the number of leaders decreases until no two leaders have the same color, after which no leader bit of any agent ever changes. 
    Afterwards, when two roots interact, they now have different colors and only one of them remains a root, due to Rule 3.
    Furthermore, note that when two leaders interact where
    one of them is a root, the one who remains a leader is also a root, due to Rule 1.
    Hence, we conclude that there is always a root, and after some number of interactions the root must be unique, after which no root bit of any agent ever changes. 
\end{proof}

Let $T$ be the number of activations described in Lemma \ref{lemma:orderinglead}. Let $L$ be the set of leaders after $T$ activations and let $q \in L$ represent the unique root. 
We now prove, by using induction on $n$, that for any integer $n, 1 \leq n < |L|$, after some number of activations $t_n \geq T$, there is a consistent linked list of $n$ links whose agents belong to $L$. 

From now on, we may refer to a leader $a \in L$ by its color $c_a$. 
Observe that, since there is only one root and no two leaders have the same color, any linked list that exists after $T$ activations, is a consistent one.


\textbf{Base case $n = 1$.}
If after $T$ activations, $q$ does not have a successor (i.e. $suc_q = c_q$), then as soon as $q$ interacts with another leader, it makes the other one its successor, due to Rule 5. 
Otherwise, as soon as $q$ interacts with $suc_q$, by Rule 5 we can be sure that they form a consistent linked list of 1 link.

\textbf{Induction step.} Suppose that $n+1 < |L|$ and after $t_{n} \geq T$ activations, a consistent linked list of $n$ links exists. Let $v_1$, $v_2$, ..., $v_n$ denote the agents succeeding agent $q$ on the linked list, respectively. Suppose $suc_{v_n} \neq c_{v_n}$, and let $p$ denote $suc_{v_n}$. Consider the first interaction between $v_n$ and $p$, after $t_{n}$ activations. After such interaction, if $pre_p = c_{v_n}$ and $d_p = d_{v_n} + 1$, we have a consistent linked list of $n+1$ links; otherwise, Rule $7$ or Rule $9$ executes and $suc_{v_n} = c_{v_n}$. 
We now assume $suc_{v_n} = c_{v_n}$.

We prove the following.
\begin{lemma}
    \label{lemma:goodagent}
    Suppose some number of activation $T' \geq T$ has passed, 
    and $q$, $v_1$, $v_2$, ..., $v_n$ form a consistent linked list 
    of $n$ links where $n + 1 < |L|$, and also $suc_{v_n} = c_{v_n}$. 
    After some more activations, a good agent exists.
\end{lemma}
\begin{proof}[Proof of Lemma \ref{lemma:goodagent}]
    If a good agent already exists after $T'$ activations, the claim is proved. 
    Therefore, we assume that no good agent exists right after $T'$ activations. 
    Define $M = L \setminus \{q, v_1, v_2, ..., v_n\}$. 
    Let $M_1$ be the set of agents in $M$ which are isolated and $M_2 = M \setminus M_1$. 
    It follows from the hypothesis that $|M| > 0$. 
    First, we prove the lemma assuming $|M_1| = |M|$. 
    Then, we prove the other cases by induction on the size of the $|M_1|$.
    
    \textbf{Case $|M_1| = |M|$.} 
    From the definitions above, it follows that $|M_1| = |M|$ implies $|M_1| > 0$ and $|M_2| = 0$. 
    Let $a \in M_1$ be an agent. 
    Since $a$ is not a good agent, we have $suc_a \neq c_a$. 
    Let $b$ denote $suc_a$. 
    Consider the first moment after $t_k$ activations in which $a$ and $b$ interact. 
    If one of $a$ or $b$ became a good agent, we are done. 
    Otherwise, Rule 4 executes and $a$ is cleared. 
    Thus, after some number of activations $t \geq T'$, $a$ is a good agent.
    
    We now use induction on $|M_1| = 0, 1, ..., |M|-1$ to prove the remaining cases.
    
    \textbf{Base case $|M_1| = 0$.} 
    Consider an agent $a_0 \in M_2$. 
    Let agent $a_1$ be $prev_{a_0}$. 
    If $a_1 \in M_2$, let agent $a_2$ be $prev_{a_1}$. 
    We repeat this process until we reach some agent $a_i$ such that either $a_i \notin M_2$ or $a_i = a_j$ for some $j < i$. Since $|M_1| = 0$, if $a_i \notin M_2$ then $a_i \notin M$ and $a_{i-1}$ is cleared by the time it and $a_i$ interact, due to Rule 8. Note that the only way $M_1$ gets new members, is that an agent becomes cleared, which implies the existence of a good agent. 
    Otherwise, $a_i = a_j$ for some $j < i$, which means that we incur in a cycle when we follow the $prev$ values of agents. In particular, there will be a pair of agents on this cycle such that when they interact (if they are not already cleared by that time), Rule 8 or Rule 9 executes and an agent is cleared. 
    Therefore, after some activations $t \geq t_n$, an agent is cleared and a good agent exists.
    
    \textbf{Induction step.} 
    Suppose $h < |M|$ and the statement holds for all $|M_1| < h$. 
    We show that it also holds for $|M_1| = h$. 
    Again, we repeat the process described in the base case. 
    This time, we stop at agent $a_i$ if any of the following holds:
    \textit{i)} $a_i \notin M$,
    \textit{ii)} $a_i = a_j$ for some $j < i$, or
    \textit{iii)} $a_i \in M_1$. 
    
    The first two cases follow from the same argument as in the base case. 
    In the third case, suppose that agents $a_{i-1}$ and $a_i$ interact at time $\interactiontime$. 
    If an agent $a \in M$ has been cleared by time $\interactiontime$, then we have a good agent. 
    Otherwise, if no agents has been cleared between $t_n$ activations and $\interactiontime$
    and, by time $\interactiontime$, agent $a_i$ is not in $M_1$ anymore, 
    then the size of $M_1$ has been reduced by at least 1. 
    The latter event implies that, by induction hypothesis, 
    after some more activations either good agent exists or, by Rule 8, 
    an interaction between $a_{i-1}$ and $a_i$ clears $a_i$. 
    Thus, eventually a good agent exists.
\end{proof}

Let $a$ be a good agent, whose existence is guaranteed by Lemma \ref{lemma:goodagent}. 
The only activation that changes the state of $a$, is an interaction with $suc_a$ (or $q$ when $suc_a = c_a$). 
If $suc_a$ is not the last agent of the linked list, it will be updated to be its successor (or $v_1$ when $suc_a = c_a$). 
Therefore, after at most $n$ such activations, $a$ interacts with the last agent on the linked list, 
and since $suc_{v_n} = c_{v_n}$, it is added to the linked list
(provided that the linked list have not already increased its size by attaching another good agent to it). 
Therefore, after some activations $t_{n+1} \geq t_{n} \geq T$, 
a consistent linked list of $n+1$ links is formed, 
concluding the induction.

We have thus proved that, after some number of activations $t_{|L|-1}$, 
there is a consistent linked list that includes all agents from $L$. 
Let $a$ be the last agent on the linked list. 
Rule 7 ensures that after some activations, $suc_a = c_a$. 
Also, after some activations all non-leader agents copy the assigned number of their leader. 
Afterwards, the whole system stabilizes and no agent changes its state, concluding the proof of the theorem.

As a final remark notice that, for an agent $a$, there may be sequences of edge activations that lead the assigned label $d_a$ to reach a value which grows as a function of $n$ before stabilizing. 
We thus assume that the variable $d_a$ \emph{overflows} when exceeding the largest number it can store, and gets set back to $0$.
Notice that $d_a$ is guaranteed to be large enough to store $k$. 
It is straightforward to verify that this latter assumption does not affect our analysis above. 
\qed

\section{Plurality Protocol with $\bigo(k^{11})$ States}

We now come back to the original problem by proving the following result. 

\begin{thm}
    \label{thm:pluralityclique}
    There is a population protocol $\pclique$ which solves the $k$-Plurality Problem under a weakly fair scheduler, when the underlying graph is complete, by employing $\bigo(k^{11})$ states per agent. 
\end{thm}

Recall that, initially, each agent $a \in V$ obtains its initial color which we shall rename to $ic_a$, from a set of possible colors $C$ of size $k$. Let $m$ be $\lceil{log_2k}\rceil$.
%
For the sake of simplicity, in this section we consider the underlying topology of possible interactions to be a complete graph. 
We show how to remove such assumption in Section General Graphs, thus proving Theorem \ref{thm:generalupper}.
A weak scheduler activates pairs of agents to interact. 
The goal is for all agents to agree on the plurality color, using as few states as possible.

\subsubsection{Main Intuition behind $\pclique$.}
The protocol proposed by Gasieniec et al. \cite{gasieniec_deterministic_2016}, which we shall call $P_r$, 
solves this problem under the hypothesis that each color is denoted by a never changing $m$-bit label, 
such that each bit is either -1 or 1, rather than the more standard 0 or 1. 
We adopt the same notation and assume that the ordering protocol $P_o$ stores the $d$ values in such format. 
The idea is to run both protocols, $P_o$ and $P_r$, in parallel and, whenever for an agent $a$, $l_a$ and $d_a$ are not equal, we ensure that after some activations, $\forall i, 0 \leq i < m: c_a[i] = w(s_a[i])$. 
When the latter condition holds, we can set $l_a$ to be $d_a$ and reinitialize $c_a$ and $s_a$ according to initialization of $P_r$. 

Notice that, since every agent is required to eventually learn the label of the plurality color, 
each agent also stores a color that corresponds to that label.
 
\textbf{Memory Organization.}
The state of each agent $a$, encodes the following information:
\begin{enumerate}
    \item $ic_a, d_a, ld_a, rt_a, pre_a$ and $suc_a$, as described for $P_o$ (where $c_a$, $l_a$ and $r_a$ in $P_o$ are renamed to $ic_a$, $ld_a$ and $rt_a$, respectively),
    \item $l_a, c_a, s_a$, as described in $P_r$, and
    \item $ans_a$, a color from the set $C$, which holds the relative majority color.
\end{enumerate}
The number of states used is at most $8k^{11}$.

\textbf{Definitions.}
An agent $a$ is called \emph{unstable}, iff $l_a \neq d_a$.
For each $i, 0 \leq i < m$, and for each $x$ that is an $i$-bit number with bit values either -1 or 1, let us define $L_x$ to be the set of all agents $a$ such that the first $i$ bits of $l_a$ are equal to $x$.

\textbf{Initialization.}
Before the execution of the protocol, for each agent $a$, 
the variables $d_a, ld_a, rt_a, pre_a$ and $suc_a$ are initialized according to $P_o$. 
Note that, instead of all bits set to 0, $d_a$ has all bits set to -1. 
Moreover, we set $l_a = d_a$ and initialize $c_a$ and $s_a$ according to $P_r$. 
$ans_a$ is set to be $ic_a$.

\textbf{Transition Function.}
We now define the transition function $\gammaclique$. 
Let us suppose that two agents $a$ and $b$ are activated, with $a \neq b$. 
Let $\Gamma_o$ be the transition function of $P_o$, and $\Gamma_r$ be the transition function of $P_r$. 

First, the values related to $P_o$ are updated according to $\Gamma_o$. 
If $l_a$ or $l_b$ is the label of the winning color in $P_2(0)$ (as described in $P_r$), 
let us set $ans_a = ic_a$ or $ans_a = ic_b$, respectively. 
Afterwards, 
\begin{enumerate}
\item If $d_a = l_a$ and $d_b = l_b$, we update the values related to $P_r$ according to $\Gamma_r$.
\item If $d_a \neq l_a$, let $L_a = \{i | 0 \leq i < m \land c_a[i] < w(s_a[i])\}$, 
and let $G_a = \{i | 0 \leq i < m \land c_a[i] > w(s_a[i])\}$. 
Let $L_b$ and $G_b$ be analogously defined. 
If $L_a \cup G_a = \emptyset$, we set $l_a = d_a$ and initialize 
$c_a$ and $s_a$ according to initialization rule of $P_r$. 
Otherwise, let $M$ be $ L_a \cup G_a$ if $d_b = l_b$, $(L_a \cup G_a) \cap (L_b \cup G_b)$ otherwise. 
For each $i \in M$, if $l_a$ and $l_b$ share the same $i$-bit prefix, we have that
\begin{enumerate}
    \item If $i \in (L_a \cap G_b) \cup (G_a \cap L_b)$, set $s_a[i] = [c_a[i]]$ and $s_b[i] = [c_b[i]]$,
    \item Otherwise, update $s_a[i]$ and $s_b[i]$ according to $\Gamma_r$.
\end{enumerate}
If $d_b = l_b$, we update the array $c_b$ and, if needed, 
we propagate the changes as in $\Gamma_r$.
\end{enumerate}

\subsubsection{Proof of Theorem \ref{thm:pluralityclique}.}

First, we prove that for each $i, 0 \leq i < m$, and for each $x$ that is an $i$-bit number with bit values equal to either -1 or 1, the two invariants of $P_2(i)$ hold for $L_x$, that is
\begin{enumerate}
    \item $\sum_{a \in L_x} c_a[i] = \sum_{a \in L_x} w(s[i])$, and
    \item $\forall a \in L_x, |w(s_a[i]) - c_a[i])| \leq 1$.
\end{enumerate}
The interactions in which states are updated according to $P_r$ satisfy the invariants due to the correctness of $P_r$. 
The other interactions can be divided into the following two cases.
\begin{enumerate}
    \item For any agent $a$, if $l_a$ changes and the change includes a bit from the $i$-bit prefix of $l_a$, 
    we know that $c_a[i] = w(s[i])$ by definition of the protocol. 
    Therefore, if $a$ is in $L_x$ before the change, the same value is subtracted from both sides of the first invariant. 
    Otherwise, if $a$ is in $L_x$ after the change, the same value is added to both sides of the first invariant. 
    Moreover, since after the reinitialization of $a$ it is still the case that $c_a[i] = w(s[i])$, the invariants hold.
    
    \item If $a \in L_x$ and $b$ are two agents such that $s_a[i]$ and $s_b[i]$ changes simultaneously, 
    then we know that $b \in L_x$ by definition of the protocol and, 
    because of the second invariant, 
    either $c_a[i] = w(s_a[i]) - 1$ and $c_b[i] = w(s_b[i]) + 1$, 
    or $c_a[i] = w(s_a[i]) + 1$ and $c_b[i] = w(s_b[i]) - 1$. 
    In both cases, $w(s_a[i]) + w(s_b[i])$ remains unchanged after $s_a[i]$ is set to $[c_a[i]]$ and $s_b[i]$ is set to $[c_b[i]]$, 
    so the first invariant holds. 
    It is immediate to check that the second invariant still holds as well.
\end{enumerate}

We proved that the two invariants hold throughout the execution of the protocol.
It follows from the correctness of $P_o$ that after some number of activations $T_o$, for each agent $a$, $d_a$ doesn't change anymore.
\begin{lemma}
    \label{lem:upper}
    After some number of activations $T \geq T_o$, for each agent $a$, $l_a$ equals $d_a$.
\end{lemma}
By the definition of $\pclique$, since $d_a$ remains unchanged after $T$ activations, it is obvious that $l_a$ remains unchanged as well. 
Thus, from the correctness of $P_r$, it follows that the whole system eventually stabilizes and every agent knows the label of the plurality color. 
Furthermore, when an agent $a$ interacts with an agent $b$ with the winning color label, $\pclique$ sets $ans_a$ to $ic_b$. Otherwise, if $a$ has the winning color itself, as soon as it is activated it sets $ans_a$ to $ic_a$ (if it is not set already). 
Therefore, it only remains to prove Lemma \ref{lem:upper}.

\subsubsection{Proof of Lemma \ref{lem:upper}.}

Suppose $T_o$ activations have passed. 
Since after $T_o$ activations, the $d$ value of agents remains unchanged, 
by the definition $\pclique$ it immediately follows that 
the number of unstable agents never increases. 

Hence, to conclude the proof it suffices to prove the following fact. 
\begin{fact}
    \label{fact:unstable}
    Suppose that, after some number of activations $T \geq T_o$ have passed, an unstable agent still exists. Then, after some additional number of activations, the number of unstable agents decreases.
\end{fact}

To see why Fact \ref{fact:unstable} holds, suppose that some number of activations $T \geq T_o$ have passed and $a$ is an unstable agent. 
Let $I = \{i | 0 \leq i < m \land c_a[i] \neq w(s_a[i]))\}$. 
Since the protocol does not change $s_a[i]$ and $c_a[i]$ for all $i \notin I$, the size of $I$ never increases. 
We prove Fact \ref{fact:unstable} by induction on $|I|$.

\textbf{Base case $|I| = 0$.}
As soon as $a$ is activated, it will set $l_a$ to $d_a$ and thus the number of unstable agents will decrease.

\textbf{Induction step.}
Suppose $|I| = n, 0 < n \leq |A|$, and for all $|I| < n$, after some activations either $I = \emptyset$ or the number of unstable agents decreases. 
Let $i \in I$ be an integer, and let $x$ denote the $i$-bit prefix of $l_a$. 
Let $U$ be the set of all agents $a \in L_x$ such that $a$ is unstable and $c_a[i] = w(s_a[i])$. 
$\pclique$ does not let agents in $U$ interact in $P_2(i)$, but any two agents from $L_x \setminus U$ can interact with each other. 
It can easily be seen that the two invariants hold for agents in $L_x \setminus U$ in $P_2(i)$. 
After some interactions, we can distinguish the following cases:  
\textit{i)} an unstable agent in $U$ becomes stable, 
\textit{ii)} an unstable agent becomes stable and is added to $L_x$, or 
\textit{iii)} the agents in $L_x \setminus U$ in $P_2(i)$ stabilize.

In the latter case, suppose without loss of generality that $c_a[i] < w(s_a[i])$.
By the first invariant of $P_2(i)$ on $L_x \setminus U$, we know that there will be another agent $b \in L_x \setminus U$ such that $c_b[i] > w(s_b[i])$. 
As soon as $a$ and $b$ interact, the protocol ensures that after the interaction, $i \notin I$. 
Thus, after some number of activations, either the number of unstable agents decreases or the size of $I$ decreases. 
Hence, Fact \ref{fact:unstable} follows by the induction hypothesis, and the proof of Lemma \ref{lem:upper} is completed. 
\qed

\section{General Graphs}
\label{sec:general}

The protocol $\pclique$ works on complete directed graphs, but it can be easily modified to work on complete undirected graphs. We now present a
protocol $\pgeneral$ which works on undirected connected graphs, under a
globally fair scheduler, and finally prove our main result, Theorem \ref{thm:generalupper}.

\subsubsection{Plurality Protocol on General Graphs.}

The idea is that, whenever a pair of agents is activated, 
the two agents can swap their updated states. 
This way, the
agents effectively \emph{travel} on the nodes of the underlying graph 
and possibly interact with other agents that were not initially adjacent.

Therefore, let us define the transition function $\gammageneral(p, q) = \gammaclique(q, p)$,
where $\gammaclique$ is the transition 
functions of modified $\pclique$.
The initialization of $\pgeneral$ is the same as that of $\pclique$

\begin{proof}[Proof of Theorem \ref{thm:generalupper}]
    Let $G$ be any connected graph. 
    Let $\initconf$, $\conf^{(1)}$
    , ... any an infinite sequence of configurations obtained by running $\pgeneral$ on $G$ under a globally fair
    scheduler, where $\initconf$ is the initial configuration. 
    Since the number of
    possible states is finite, the number of possible configurations is also
    finite. 
    Therefore, there exists a configuration $\conf$ that appears
    infinitely often in the sequence. 
    
    For all distinct pairs $u,v \in V(G)$,
    let $\gpath{u}{v}$ be a series of edges forming a path from $u$ to $v$, 
    and suppose that the edges in
    $Path_{u,v}$ gets activated first in the order in which they appear in the path, and then in reverse order. 
    Let $\sact u v$ be the concatenation of such edge activations. 
    If we activate edges according to
    $\sact u v$, then $u$ travels along $\gpath{u}{v}$ (possibly interacting
    with some other agents), until it interacts with $v$, and then travels back to
    its position. 
    Therefore, the sequence of activations $\sact u v$ ensures that pair $\{u,v\}$ of agents interact with each other at least once. 
    If we keep activating edges according to the sequences $\{\sact uv\}_{u,v\in V}$, for each pair of agents $\{u,v\}$, 
    then starting from $\conf$, 
    each pair of agents interact infinitely often. 
   
    Remark that, a globally fair scheduler is also a weakly fair one. By correctness of $\pclique$ under
    a weakly fair scheduler (Theorem \ref{thm:pluralityclique}), by repeating the mentioned edge activation sequence
    starting from $\conf$, a \emph{stable}
    configuration $\conf'$ will be reached (a configuration in which all agents know the initial
    plurality color, and their guess remains correct thereafter).
    Therefore, $\conf'$ is reachable from $\conf$.
    By the definition of a globally fair scheduler, since $\conf$ is infinitely reached, the stable configuration $\conf'$
    is eventually reached.
    
\end{proof}

\end{document}